\documentclass[english,review]{elsarticle}
\usepackage[T1]{fontenc}
\usepackage[latin9]{inputenc}
\usepackage{amsthm}
\usepackage{amsmath}
\usepackage{amssymb}
\usepackage{graphicx}

\makeatletter
\numberwithin{equation}{section}
\numberwithin{figure}{section}
\theoremstyle{plain}
\newtheorem{thm}{\protect\theoremname}

\pdfoutput=1
\journal{Journal of Economic Behavior and Organization}

\makeatother

\usepackage{babel}
\providecommand{\theoremname}{Theorem}

\begin{document}
\begin{frontmatter}

\title{Specialists and Generalists: Equilibrium Skill Acquisition Decisions
in Problem-solving Populations}

\author[rvt]{Katharine A. Anderson\corref{cor1}}

\ead{andersok@andrew.cmu.edu}

\cortext[cor1]{Corresponding author: +1-(412) 427-1904}

\address[rvt]{Tepper School of Business, Carnegie Mellon University, Pittsburgh,
PA 15213, USA}
\begin{abstract}
Many organizations rely on the skills of innovative individuals to
create value, including academic and government institutions, think
tanks, and knowledge-based firms. Roughly speaking, workers in these
fields can be divided into two categories: specialists, who have a
deep knowledge of a single area, and generalists, who have knowledge
in a wide variety of areas. In this paper, I examine an individual's
choice to be a specialist or generalist. My model addresses two questions:
first, under what conditions does it make sense for an individual
to acquire skills in multiple areas, and second, are the decisions
made by individuals optimal from an organizational perspective? I
find that when problems are single-dimensional, and disciplinary boundaries
are open, all workers will specialize. However, when there are barriers
to working on problems in other fields, then there is a tradeoff between
the depth of the specialist and the wider scope of problems the generalist
has available. When problems are simple, having a wide variety of
problems makes it is rational to be a generalist. As these problems
become more difficult, though, depth wins out over scope, and workers
again tend to specialize. However, that decision is not necessarily
socially optimal--on a societal level, we would prefer that some workers
remain generalists. \end{abstract}
\maketitle
\emph{Thanks to Scott Page and Ross O'Connell. This work was supported
by the NSF and the Rackham Graduate School, University of Michigan.
Computing resources supplied by the Center for the Study of Complex
Systems, University of Michigan.}

\end{frontmatter}

\section{Introduction}

Many organizations rely on the skills of innovative individuals to
create value. Examples include academic institutions, government organizations,
think tanks, and knowledge-based firms. Workers in these organizations
apply a variety of skills to in order to solve difficult problems:
architects design buildings, biochemists develop new drugs, aeronautical
engineers create bigger and better rockets, software developers create
new applications, and industrial designers create better packaging
materials. Their success--and thus the success of the organizations
they work for--is dependent on the particular set of skills that they
have at their disposal, but in most cases, the decision of which skills
to acquire is made by individuals, rather than organizations. The
perception is that these workers choose to become more specialized
as the problems they face become more complex (Strober (2006)) . This
perception has generated a countervailing tide of money and institutional
attention focused on promoting interdisciplinary efforts. However,
we have very little real understanding of what drives an individual's
decision to specialize. 

Roughly speaking, workers in knowledge-based fields can be divided
into two categories: specialists, who have a deep knowledge of a single
area, and generalists, who have knowledge in a wide variety of areas.%
\footnote{This dichotomy is often summed up in the literature via a metaphor
used by Isaiah Berlin in an essay on Leo Tolstoy: {}``The fox knows
many things, but the hedgehog knows one big thing'' (Berlin (1953))
In other words, foxes are generalists with a wide variety of tools
to apply to problems (albeit sometimes inexpertly) and hedgehogs are
specialists who have a single tool that they can apply very well. %
} In this paper, I consider an individual's decision to be a specialist
or a generalist, looking specifically at two previously unaddressed
questions. First, under what conditions does it make sense for an
individual to acquire skills in multiple areas? And second, are the
decisions made by individuals optimal from an organizational perspective? 

Most of the work done on specialists and generalists is focused on
the roles the two play in the economy. Collins (2001) suggests that
specialists are more likely to found successful companies. Lazear
(2004 and 2005), on the other hand, suggests that the successful entrepreneurs
should be generalists--a theory supported by Astebro and Thompson
(2011), who show that entrepreneurs tend to have a wider range of
experiences than wage workers. Tetlock (1998) finds that generalists
tend to be better forecasters than specialists. In contrast, a wide
variety of medical studies (see, for example, Hillner et al (2000)
and Nallamothu et al (2006)), show that outcomes tend to be better
when patients are seen by specialists, rather than general practitioners.
However, none of this work considers the decision that individuals
make with respect to being a specialist or generalist. While some
people will always become generalists due to personal taste, the question
remains: is it ever rational to do so in the absence of a preference
for interdisciplinarity? And is the decision that the individual makes
optimal from a societal perspective?

There is evidence that being a generalist is costly. Adamic et al
(2010) show that in a wide variety of contexts, including academic
research, patents, and contributions to wikipedia, the contributions
of individuals with greater focus tend to have greater impact, indicating
that there is a tradeoff between the number of fields an individual
can master, and her depth of knowledge in each. This should not be
surprising. Each of us has a limited capacity for learning new things--by
focusing on a narrow field of study, specialists are able to concentrate
their efforts and maximize the use of that limited capacity, while
generalists are forced to spread themselves more thinly in the pursuit
of a wider range of knowledge. In the language of economics, generalists
pay a fixed cost for each new field of study they pursue, in the form
of effort expended learning new jargon, establishing new social contacts
in a field, and becoming familiar with new literatures.

Given that it is costly to diversify ones skills, the decision to
become a generalist can be difficult to rationalize. In this paper,
I examine model in which workers decide whether to be specialists
or generalists to explore conditions under which it is rational for
an individual to choose to be a generalist. I show that when problems
are single-dimensional and there are no barriers to working on problems
in other disciplines, the equilibrium population contains only specialists.
However, when there are barriers to working on problems in other fields
(eg: communication barriers or institutional barriers) then there
is a tradeoff between the depth of study of the specialist and the
wide scope of problems that the generalist can work on. When problems
are relatively simple, generalists dominate because their breadth
of experience gives them a wider variety of problems to work on. But
as problems become more difficult, depth wins out over scope, and
workers tend to specialize. 

I then show that the equilibrium decisions reached by individuals
are not necessarily socially optimal. As problems become harder, individual
workers are more likely to specialize, but as a society, we would
prefer that some individuals remain generalists. This disconnect reflects
the fact that from a social perspective, we would prefer to have researchers
apply the widest possible variety of skills to the problems we face,
but individuals internalize the cost of obtaining those skills. Thus,
the model predicts that some populations will suffer from an undersupply
of generalists. In such populations, it would be socially beneficial
to subsidize the acquisition of skills in broad subject areas. 

Finally, I consider an extension of the model in which problems have
multiple parts. This allows me to consider problems that are explicitly
multidisciplinary--that is, when different parts of a problem are
best addressed using skills from different disciplines. I show when
problems are multidisciplinary, it is possible to rationalize being
a generalist, even when there are no disciplinary boundaries. In particular,
when there is a large advantage to applying the best tool for the
job, being a generalist is optimal.

\section{Model}

I construct a two period model. In period 1, the workers face a distribution
of problems and each worker chooses a set of skills. In period 2,
a problem is drawn from the distribution, and the workers attempt
to solve it using the skills they acquired in period 1. I will solve
for the equilibrium choice of skills in period 1.

Let $S$ be the set of all possible skills.%
\footnote{Skills are defined as bits of knowledge, tools, and techniques useful
for solving problems and not easily acquired in the short run. See
Anderson (2010) for a model with a similar treatment of skills.%
} The skills are arranged into 2 disciplines, $d_{1}$ and $d_{2}$,
each with $K$ skills, $s_{1d}...s_{Kd}$. An example with six skills
arranged into two disciplines is shown in Figure \ref{fig:Two-disciplines}.
\begin{figure}
\includegraphics[clip,width=0.8\columnwidth]{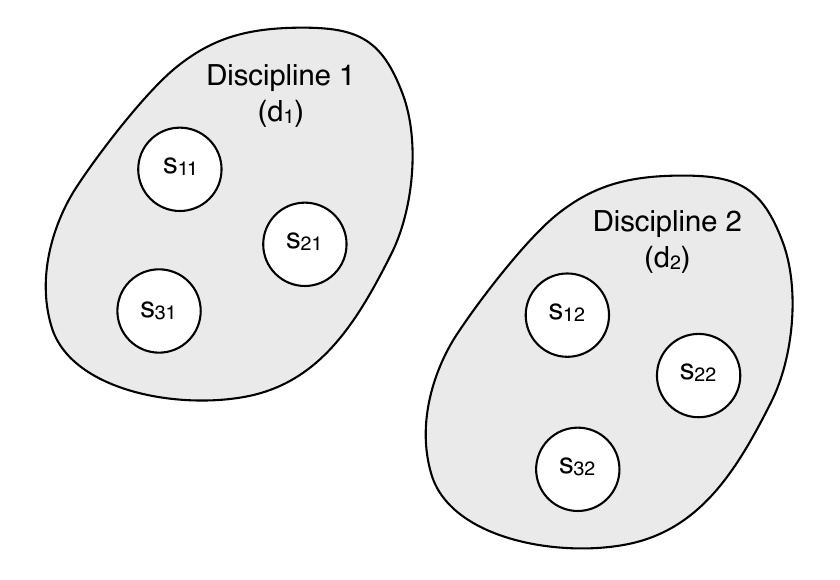}

\caption{\label{fig:Two-disciplines}Two disciplines, each with three skills.}
\end{figure}
A \emph{specialist} is a person who chooses skills within a single
discipline. A \emph{generalist }is a person who chooses some skills
from both disciplines.

A problem, $y$, is a task faced by the workers in the model. A \emph{skill}
is a piece of knowledge that can be applied to the problem in an attempt
to solve it. Each skill $s_{kd}\in S$ has either a high probability
$\left(H\right)$ or a low probability $\left(L\right)$ of solving
the problem.\emph{ }I will define a problem by the matrix of probabilities
that each skill will solve the problem. That is, $y=\left[\begin{array}{ccc}
y_{11} &  & y_{12}\\
\vdots &  & \vdots\\
y_{K1} &  & y_{K2}
\end{array}\right]$ where $y_{kd}=H$ if skill $k$ in discipline $d$ has a high probability
of solving the problem and $L$ if it has a low probability of solving
the problem. So, for example, if there are two disciplines, each with
three skills, a problem might be
\[
y=\left[\begin{array}{cc}
L & H\\
H & H\\
H & L
\end{array}\right]
\]
meaning that two of the skills in each discipline have a high probability
of solving the problem, and one skill in each discipline has a low
probability of solving the problem. Define $h\equiv1-H$ and $l\equiv1-L$.

The mechanics of the model are as follows. In period 1, the workers,
$i_{1}...i_{N}$, each choose a set of skills $A_{i}\subset S$. In
period 2, the workers attempt to solve a problem using those skills.
I will assume that workers have a \emph{capacity }for learning skills,
which limits the number of skills they can obtain. In the current
context, I will assume that all workers all have the same capacity
for learning new skills, and that all skills are equally costly to
obtain.%
\footnote{The case where workers have different capacities would obviously be
an interesting extension, as would the case where different skills
had different costs. %
} Let $M\in\mathbb{Z}^{+}$ represent an individual's capacity for
new skills and let $q=1$ be the cost of acquiring a new skill. I
assume that workers pay a fixed cost, $c$, for learning skills in
a new discipline. That is, a worker pays $1+c$ to obtain the first
skill in a discipline, and $q=1$ for every additional skill in that
discipline. For simplicity, I will assume that $M=K+c$. This assumption
means that a specialist can obtain all $K$ skills in one discipline,
and a generalist can obtain a total of $K-c$ skills spread over the
two disciplines. 

Although workers in period 1 do not know the particular problem they
will face in period 2, they do know the distribution, $\Delta$, from
which those problems will be drawn. In particular, they know the probability
that each skill will be an H skill or an L skill. For simplicity,
I will make two assumptions about the distribution of problems: 1)
skills are \emph{independent, }meaning that the probability that skill
$s_{kd}$ is an H skill is independent of the probability that skill
$s_{k'd'}$ is an H skill%
\footnote{This assumption means that skills must be applied more or less independently.
That is, it cannot be the case that skills are used in combination
to solve problems, or that skills build on one another.%
} and 2) skills are \emph{symmetric within disciplines}, meaning that
every skill in a discipline has an equal probability of being an $H$
skill.%
\footnote{This assumption simplifies the decision making process for generalists.
When skills are symmetric within a discipline, a generalist's skill
acquisition decision is simply a division of her skills across the
two disciplines--within a discipline, she can choose her skills at
random. %
} 

This knowledge of the distribution of problems can be translated into
knowledge about individual skills. Let $\delta_{d}$ be the probability
that a skill in discipline $d$ is an $H$ skill--that is, $\delta_{d}=E\left[Prob\left(y_{kd}=H\right)\right]$
where the expectation is taken over the distribution of problems,
$\Delta$. The vector of probabilities in the two disciplines, $\delta=\left[\delta_{1},\delta_{2}\right]$,
is known \emph{ex ante.} 

Workers choose their skills in period 1 to maximize their expected
probability of solving the problem in period 2. A Nash equilibrium
of this game is a choice of skill set for each worker in the population,
$A=\left\{ A_{1}...A_{N}\right\} ,$ such that no worker has an incentive
to unilaterally change her skill set, given the distribution of problems.

\section{\label{sec: Specialization and Barriers Between Disciplines}Results:
Specialization and Barriers Between Disciplines}

In this section, I consider two questions. The first question concerns
individual decision-making--what is the equilibrium skill acquisition
decision of the workers? Under what conditions do individuals decide
to generalize? The second question concerns the optimality of that
population from an organizational perspective. Is the equilibrium
population optimal?

Note that in order to simplify the exposition, I will consider a special
case where all disciplines are equally useful in expectation--that
is, where $\delta_{1}=\delta_{2}=\delta$. It is straightforward to
generalize the results to a case where $\delta_{1}\ne\delta_{2}$
(see Appendix for the details).

\subsection{Equilibrium Skill Populations with No Barriers between Disciplines}

Given that generalists pay a significant penalty for diversifying
their skills, it is difficult to explain the existence of generalists
in the population. Theorem \ref{thm:If Open No Generalists} states
that if workers can work on any available problem, then there will
be no generalists in equilibrium.
\begin{thm}
\label{thm:If Open No Generalists}If skills are independent and symmetric
within discipline, and workers can work on any available problem,
then no worker will ever want to be a generalist and the equilibrium
population will contain only specialists.\end{thm}
\begin{proof}
The \emph{ex ante }probability that a specialist in discipline $i$
will be able to solve a problem from a given distribution, $\Delta$,
is
\begin{eqnarray*}
E\left[P\left(S_{i}\right)\right] & = & \sum_{y}Prob\left(\mbox{one of skills solves }y\right)*\Delta\left(y\right)\\
 & = & \sum_{y}\left(1-Prob\left(\mbox{none do}\right)\right)*\Delta\left(y\right)\\
 & = & 1-\sum_{n_{i}=0}^{K}h^{n_{i}}l^{K-n_{i}}\left({K\atop n_{i}}\right)\delta^{n_{i}}\left(1-\delta\right)^{K-n_{i}}\\
 & = & 1-\left(\delta h+\left(1-\delta\right)l\right)^{K}
\end{eqnarray*}
where $n_{i}$ is the number of $H$ skills in discipline $i$ in
a particular problem, $y$. 

Now, consider a generalist who is spreading his skills across both
disciplines. The \emph{ex ante }probability that a generalist with
$x$ skills in discipline $1$ and $K-c-x$ skills in discipline 2
will solve a problem from a given distribution, $\Delta$, is
\begin{eqnarray*}
E\left[P\left(G\right)\right] & = & 1-\sum_{y}Prob\left(\mbox{none of skills solve }y\right)*\Delta\left(y\right)\\
 & = & 1-\left(\delta h+\left(1-\delta\right)l\right)^{K-c}
\end{eqnarray*}

$1-\left(\delta h+\left(1-\delta\right)l\right)^{K}>1-\left(\delta h+\left(1-\delta\right)l\right)^{K-c}$,
and thus no individual will ever be a generalist in two disciplines.
(See Appendix for the same result with $\delta_{1}\ne\delta_{2}$)
\end{proof}
Note that this result generalizes to a case with more than two disciplines.
Generalists do worse as they add skills in additional disciplines,
so this result holds regardless of the number of disciplines a generalist
spreads himself across.

\subsection{Equilibrium Skill Populations with Barriers between Disciplines}

Theorem \ref{thm:If Open No Generalists} clearly indicates that when
workers can solve problems in other fields, there is no advantage
to being a generalist. However, in practice, there may be many barriers
between disciplines that prevent a worker in one discipline from solving
problems in another. Cultural or institutional barriers may prevent
her from working on questions in other disciplines, either because
resources are not forthcoming or because it is difficult to get compensated
for work in other areas. Communication barriers are also a significant
impediment to interdisciplinary work--although a software engineers
may have skills useful in solving user interface problems, field-specific
jargon may make it difficult for her to communicate her insights.
If communication barriers are severe enough, she may even have difficulty
understanding what open questions exist. Finally, a person in one
field may simply be unaware of problems that exist in other fields,
even if her skills would be useful in solving them. 

Barriers to working on problems outside ones discipline give us the
ability to talk about the {}``scope'' of a worker's inquiry. Generalists
are able to work on a broader set of problems, and thus their scope
is larger than that of specialists. There is therefore a tradeoff
between the depth of skill gained through specialization and the scope
gained through generalization. A specialist has a depth of skill that
gives her a good chance of solving the limited set of problems in
the area she specializes in. Generalists have a limited number of
skills, but are able to apply those skills to a much broader set of
problems. Thus, the choice between being a specialist and a generalist
can be framed in terms of a tradeoff between the depth of one's skill
set and scope of one's problem set. 

More formally, choice of whether to specialize will depend on two
parameters. First, let $\pi\left(\delta,h,l\right)\equiv\left(\delta h+\left(1-\delta\right)l\right)$
be the expected probability that a skill won't be able to solve a
problem drawn from $\Delta$. When $\pi$ is large, the probability
that any one skill will solve the problem is very low. Thus, we can
think of problems becoming more difficult as $\pi$ increases. Second,
let $\phi$ be the fraction of all problems that occur in discipline
1. When $\phi$ is very large or very small, most of the problems
fall in one field or another, limiting the value of increasing the
scope of the problem set. 

These two parameters--$\phi$ and $\pi$--define a range in which
workers will choose to generalize in equilibrium. This range is illustrated
in Figure \ref{fig:Generalists Range}. This diagram illustrates the
tradeoff between depth and scope. When problems are easy to solve,
scope is more valuable than depth. However, as $\pi$ increases and
problems become more difficult, depth wins out over scope, and the
range in which individuals choose to generalize shrinks. 

\begin{figure}
\includegraphics[width=0.9\columnwidth]{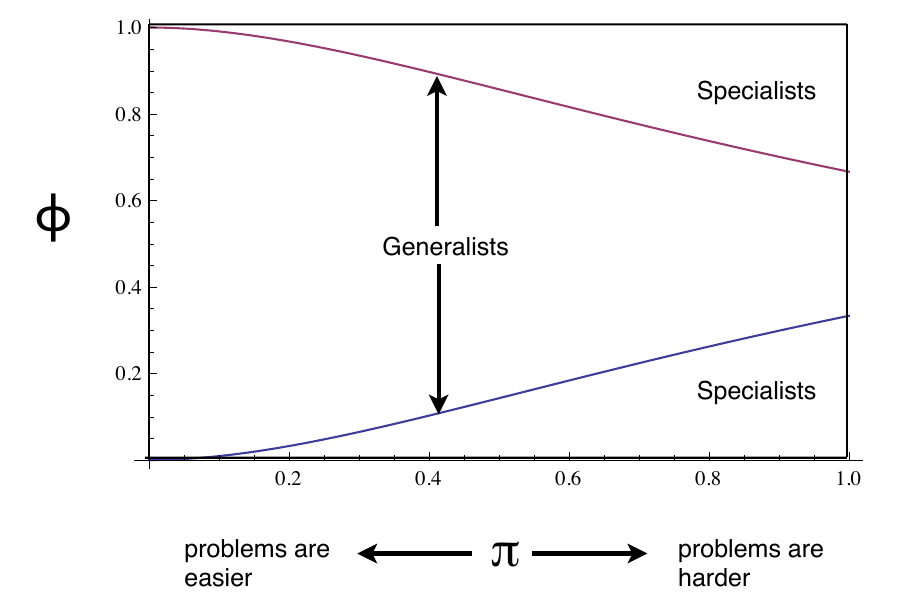}

\caption{\label{fig:Generalists Range}Equilibrium skill acquisition decisions
when $k=3$ and $c=1$. }

\end{figure}

Theorem \ref{thm:Communication Barriers and Generalists} summarizes
these results.
\begin{thm}
\label{thm:Communication Barriers and Generalists}If skills are independent
and symmetric within discipline, and there are barriers to working
on problems in other disciplines, then workers will generalize if
$1-\left(\frac{1-\pi^{K-c}}{1-\pi^{K}}\right)\le\phi\le\frac{1-\pi^{K-c}}{1-\pi^{K}}$
where $\phi$ is the fraction of problems assigned to discipline 1.
If $\phi>\frac{1-\pi^{K-c}}{1-\pi^{K}}$, then workers will all specialize
in discipline 1 and if $\phi<1-\left(\frac{1-\pi^{K-c}}{1-\pi^{K}}\right)$
then workers will all specialize in discipline 2.\end{thm}
\begin{proof}
In this case, the \emph{ex ante} probability that a problem is solved
by a specialist is $\phi\left(1-\left(\delta h+\left(1-\delta\right)l\right)^{K}\right)$
for a specialist in discipline 1 and $\left(1-\phi\right)\left(1-\left(\delta h+\left(1-\delta\right)l\right)^{K}\right)$
for a specialist in discipline 2. Since generalists can work on problems
in both disciplines, their expected probability of solving the problem
is $1-\left(\delta h+\left(1-\delta\right)l\right)^{K-c}$. A worker
will generalize if $E\left[P\left(S_{1}\right)\right]<E\left[P\left(G\right)\right]$
and $E\left[P\left(S_{2}\right)\right]<E\left[P\left(G\right)\right]$.
The result follows immediately. (See Appendix for the same result
with $\delta_{1}\ne\delta_{2}$)
\end{proof}
Note that the size of the regions in which workers specialize depends
on how costly it is to diversify ones skills. As the fixed cost of
learning something in a new discipline increases ($c\uparrow$), the
regions in which people specialize grow.

\subsection{Optimality of the Equilibrium}

In this section, I consider whether this distribution of specialists
and generalists in the population is optimal, from a societal perspective.
There is reason to believe that it would not be. From a societal standpoint,
we would like to maximize the probability that someone manages to
solve the problem. This means that as a society, we would prefer to
have problem solvers apply as wide a range of skills as possible.
But workers who diversify their skills obtain fewer skills overall,
which tends to make the individual want to specialize. The result
of this disconnect between individual and social welfare is a range
in which generalists are under provided (see Figure \ref{fig:Regions of Social Suboptimality}).
As problems become more difficult, this region of suboptimality grows,
as is illustrated in Figure \ref{fig:Regions of Social Suboptimality Grow}. 

\begin{figure}
\includegraphics[clip,width=0.9\columnwidth]{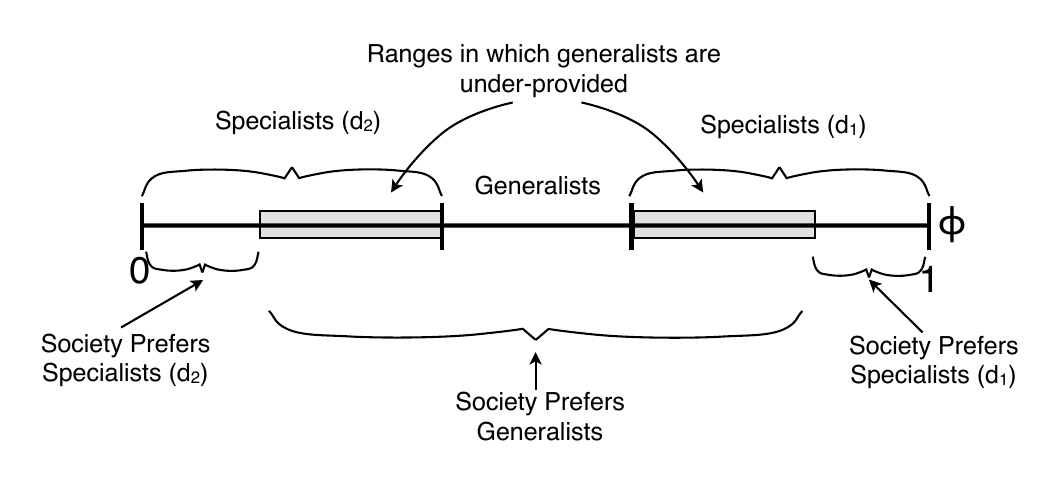}

\caption{\label{fig:Regions of Social Suboptimality}}
\end{figure}

\begin{figure}
\includegraphics[width=1\columnwidth]{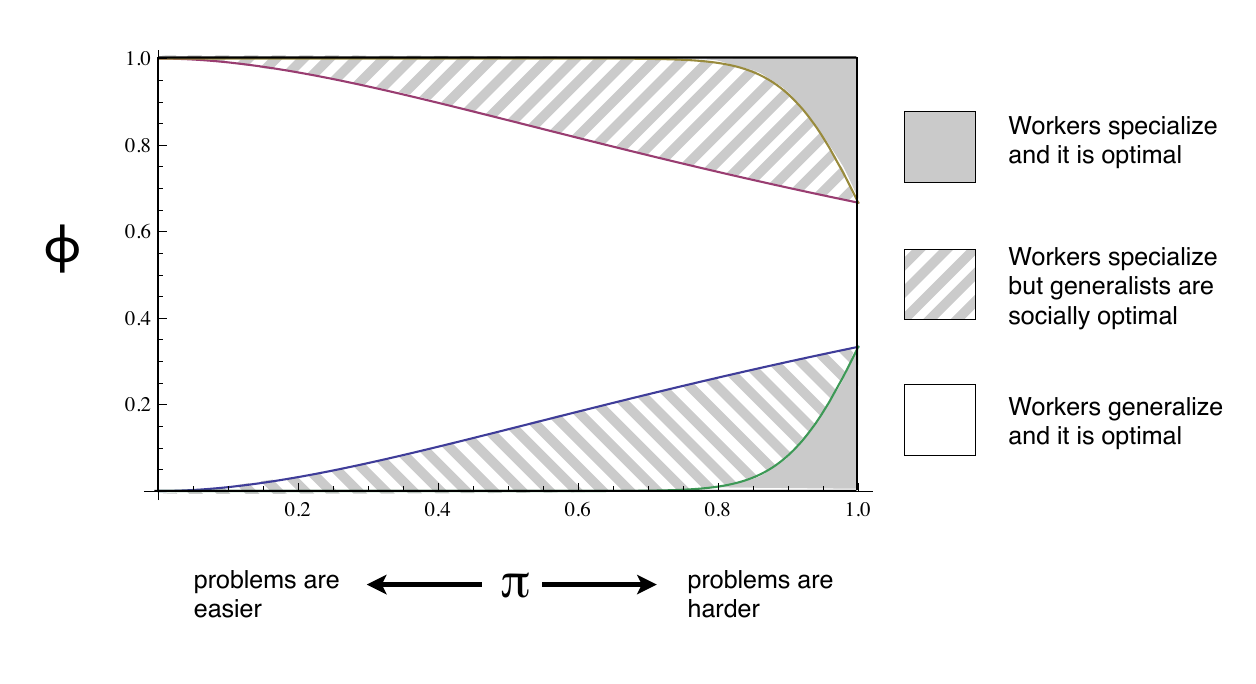}

\caption{\label{fig:Regions of Social Suboptimality Grow}Regions of social
suboptimality for $k=3$, $ $$c=1$, $N=10$ }
\end{figure}

Theorem \ref{thm: Optimality of the Equilibrium} summarizes these
results.
\begin{thm}
\label{thm: Optimality of the Equilibrium}If skills are independent
and symmetric within discipline, and there are barriers to working
on problems in other disciplines, then there is a range of values
for $\phi$ (the fraction of problems assigned to discipline 1) such
that generalists are underprovided in the equilibrium population of
problem solvers. 

In particular, generalists are underprovided when $\frac{1-\pi^{K-c}}{1-\pi^{K}}<\phi<\frac{1-\pi^{N\left(K-c\right)}}{1-\pi^{NK}}$
or $1-\frac{1-\pi^{N\left(K-c\right)}}{1-\pi^{NK}}<\phi<1-\frac{1-\pi^{K-c}}{1-\pi^{K}}$. \end{thm}
\begin{proof}
The probability that at least one of the $N$ problem-solvers in the
population solves the problem is $1-Prob\left(\mbox{none of them do}\right)$.
If all of the individuals in the population are specialists in discipline
1, then with probability $\phi$, each specialist has a probability
$1-\pi^{K}$ of solving the problem and $\pi^{K}$ of not solving
it. With probability $1-\phi$, the problem is assigned to the other
discipline, and no specialist solves it. Thus, the probability of
someone in a population of discipline 1 specialists solving the problem
is 
\begin{eqnarray*}
Prob\left(\mbox{one of N solve it}\right) & = & 1-Prob\left(\mbox{none of N solve it}\right)\\
 & = & 1-\left[\phi Prob\left(\mbox{none solve problem in }d_{1}\right)\right.\\
 &  & \left.+\left(1-\phi\right)Prob\left(\mbox{none solve problem in }d_{2}\right)\right]\\
 & = & 1-\left[\phi Prob\left(\mbox{one fails}\right)^{N}+\left(1-\phi\right)*1\right]\\
 & = & 1-\left[\phi\left(\pi^{K}\right)^{N}+\left(1-\phi\right)*1\right]\\
 & = & \phi\left(1-\pi^{KN}\right)
\end{eqnarray*}
On the other hand, if they are all generalists, then the probability
of at least one solving the problem is 
\begin{eqnarray*}
Prob\left(\mbox{one of N solve it}\right) & = & 1-Prob\left(\mbox{none of N solve it}\right)\\
 & = & 1-\left(\pi^{K-c}\right)^{N}\\
 & = & 1-\pi^{N\left(K-c\right)}
\end{eqnarray*}

Society is better off with a population of generalists when $1-\pi^{N\left(K-c\right)}>\phi\left(1-\pi^{KN}\right)$,
which is true when $\phi<\frac{1-\pi^{N\left(K-c\right)}}{1-\pi^{NK}}$.
However, there is a population of generalists when $\phi\le\frac{1-\pi^{K-c}}{1-\pi^{K}}$
. It is always the case that $\frac{1-\pi^{K-c}}{1-\pi^{K}}\le\frac{1-\pi^{N\left(K-c\right)}}{1-\pi^{NK}}$.
So if $\frac{1-\pi^{K-c}}{1-\pi^{K}}<\phi<\frac{1-\pi^{N\left(K-c\right)}}{1-\pi^{NK}}$,
then society is better off with a population of generalists, but has
a population of specialists. 

We can make a similar argument for specialists in discipline 2. Society
is better off with a population of generalists when $1-\pi^{N\left(K-c\right)}>\left(1-\phi\right)\left(1-\pi^{KN}\right)$,
which is true when $\phi>1-\frac{1-\pi^{N\left(K-c\right)}}{1-\pi^{NK}}$.
However, there is a population of generalists when $\phi>1-\frac{1-\pi^{K-c}}{1-\pi^{K}}$
. It is always the case that $1-\frac{1-\pi^{N\left(K-c\right)}}{1-\pi^{NK}}\le1-\frac{1-\pi^{K-c}}{1-\pi^{K}}$.
So if $1-\frac{1-\pi^{N\left(K-c\right)}}{1-\pi^{NK}}<\phi<1-\frac{1-\pi^{K-c}}{1-\pi^{K}}$,
then society is better off with a population of generalists, but has
a population of specialists. (See Appendix for the same result with
$\delta_{1}\ne\delta_{2}$)
\end{proof}
Note that the size of the regions of suboptimality will depend on
the number of individuals in the population. As $N$ increases, the
suboptimal regions become larger.

\section{An Extension: Problems with Multiple Parts}

In the previous section, I showed that barriers to addressing problems
in other disciplines can induce problem solvers to diversify their
skills. In this section, I consider an extension of the previous model,
which highlights a second scenario in which individuals can be incentivized
to acquire skills in multiple disciplines: problems with multiple
parts. As problems become increasingly complicated, they may be broken
down into many different sub-problems. Although in some cases, these
subproblems may all be best addressed within a single discipline,
in others, different subproblems will be best addressed using different
skills. In this section, I show that when problems are \emph{multidisciplinary}--that
is, when different parts of a problem are best addressed using different
disciplines--then a population of generalists can be sustained.

\subsection{Problems With Multiple Parts}

As in the previous model, skills in the set $S$ are divided into
two disciplines, $d_{1}$ and $d_{2}$. Workers use their skills to
address a problem, the nature of which is not known \emph{ex ante.
}They will choose to be a specialist or generalist in period 1 to
maximize their chances of solving the problem in period 2.\emph{ }But
now, suppose each problem consists of two parts, $y^{1}$ and $y^{2}$.
In order to solve the problem, an individual must solve all parts
of the problem.%
\footnote{This is essentially an adaptation of Kremer's O-ring Theory (Kremer
(1993)).%
} Each part of the problem is addressed independently by the skills
in each of the disciplines. Thus, much as before, we can define the
parts of the problem by a matrix of probabilities that each skill
will solve the problem. That is, $y^{i}=\left[\begin{array}{ccc}
y_{11}^{i} &  & y_{12}^{i}\\
\vdots &  & \vdots\\
y_{K1}^{i} &  & y_{K2}^{i}
\end{array}\right]$ where $y_{kd}^{i}=H$ if skill $k$ in discipline $d$ has a high
probability of solving part $i$ and $L$ if it has a low probability
of solving part $i$.

As in the previous section, I will assume that for each part of the
problem, skills are independent\emph{ }($Prob(y_{kd}^{i}=H)$ uncorrelated
with $Prob(y_{k'd'}^{i}=H)$) and skills are symmetric within disciplines
($Prob(y_{kd}^{i}=H)=Prob(y_{jd}^{i}=H)$) . 

As before, the probability that a given skill is an $H$ skill is
not known \emph{ex ante.} However, the workers know the expected probability
that a skill is an $H$ skill. I will allow the expected probabilities
to vary across parts of the problem--in other words, it is possible
that a discipline will be more useful in solving one of the parts
of the problem than in solving the other part of the problem. Let
$\delta_{d}^{i}$ be the probability that a skill from discipline
$d$ is an H skill for part $i$ of the problem. That is, $\delta_{d}^{i}=E\left[Prob\left(y_{kd}^{i}=H\right)\right]$.
The matrix $\delta=\left[\begin{array}{cc}
\delta_{1}^{1} & \delta_{1}^{2}\\
\delta_{2}^{1} & \delta_{2}^{2}
\end{array}\right]$ describes a distribution of problems, $\Delta$, and is known \emph{ex
ante. }An entry in the $i^{th}$ column of that matrix is the vector
of probabilities that a skill in each of the disciplines will be useful
for solving part $i$ of the problem. 

We can categorize the problems according to the relative usefulness
of the two disciplines in the two parts of the problem. There are
two categories for the problems:
\begin{enumerate}
\item One discipline is as or more useful for both parts of the problem:
$\delta_{1}^{i}\ge\delta_{2}^{i}\forall i$ 
\item One discipline is more useful for part 1 and the other discipline
is more useful for part 2: $\delta_{1}^{i}>\delta_{2}^{i}$ and $\delta_{1}^{j}<\delta_{2}^{j}$ 
\end{enumerate}
These categories are illustrated in Figure \ref{fig:Taxonomy of Distributions}. 

\begin{figure}
\includegraphics[clip,width=1\columnwidth]{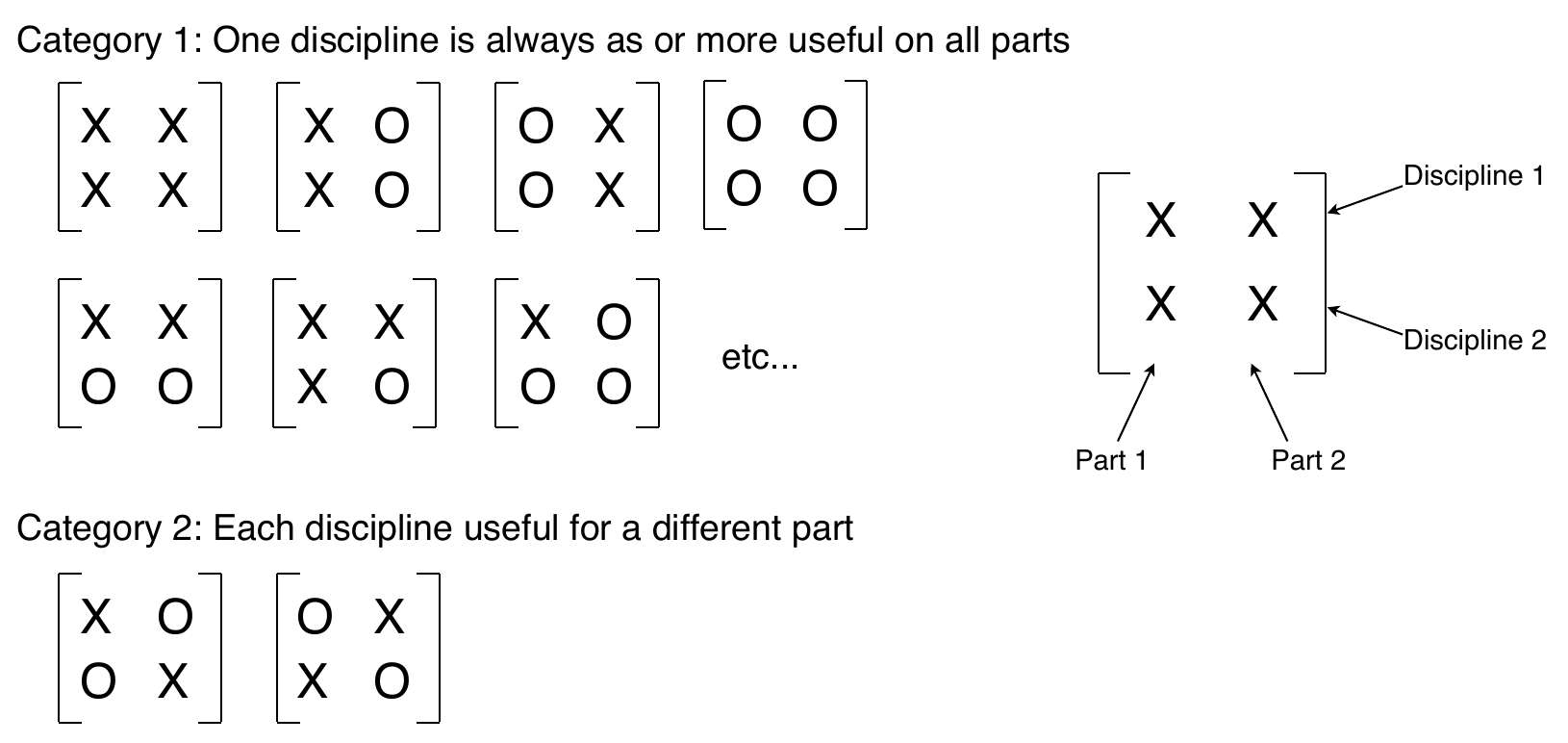}

\caption{\label{fig:Taxonomy of Distributions}}
\end{figure}

If a problem falls into the first category, then the results are similar
to those obtained in Section \ref{sec: Specialization and Barriers Between Disciplines}.
In particular, if there are no barriers to working on problems in
other disciplines, then all workers will specialize. If a problem
falls into the second category, then the results do not resemble any
of those already explored. Problems with multiple parts, each of which
is best addressed within the context of a different discipline, are
often referred to as \emph{multidisciplinary. }Generalists have an
advantage in multidisciplinary problems, because they can apply different
types of skills to different parts of a problem. For example, suppose
a scientist is look at nerve conduction in an organism. That problem
may have elements are are best addressed using biological tools, and
other elements that are best addressed using physics tools. An individual
with both biology and physics skills will have an advantage over someone
who is forced to use (for example) physics skills to solve both parts
of the problem. The below states that when problems are multidisciplinary,
it can be rational to be a generalist, even in the absence of barriers
to working in other fields. 

More formally, suppose that if a worker uses the {}``right'' discipline
for a part of a problem, then there is a probability $\delta_{1}$
that a skill in that discipline is useful ($\delta_{1}=Prob\left(y_{kd}^{i}=H\right)$
when $d$ is the right discipline to use for part $i$ of the problem).
If she uses the {}``wrong'' discipline, then there is a probability
$\delta_{0}$ that a skill in that discipline is useful ($\delta_{1}=Prob\left(y_{kd}^{i}=H\right)$
when $d$ is the wrong discipline to use for part $i$ of the problem).
This is without loss of generality, because the only thing that makes
a problem multidisciplinary is the ordering of the usefulness of the
disciplines. Further, let $\pi_{1}=\delta_{1}h+\left(1-\delta_{1}\right)l$
and $\pi_{0}=\delta_{0}h+\left(1-\delta_{0}\right)l$. These represent
the probability that a skill in the right discipline will not solve
a part of a problem and the probability that a skill in the wrong
discipline will not solve part of the problem. Note that $\pi_{1}<\pi_{0}$. 

When the efficacy of the two disciplines is very different ($\pi_{1}\ll\pi_{0}$),
then using the right skill for the job has a large effect on the probability
of solving the problem as a whole, and it will be rational to obtain
skills in multiple disciplines. Figure \ref{fig:Regions of Specialization in Two Part Problem}
illustrates the region in which individuals choose to be generalists
and specialists, and Theorem \ref{Thm: Multidisciplinary Problems}
summarizes the result.
\begin{figure}[h]
\includegraphics[width=0.9\columnwidth]{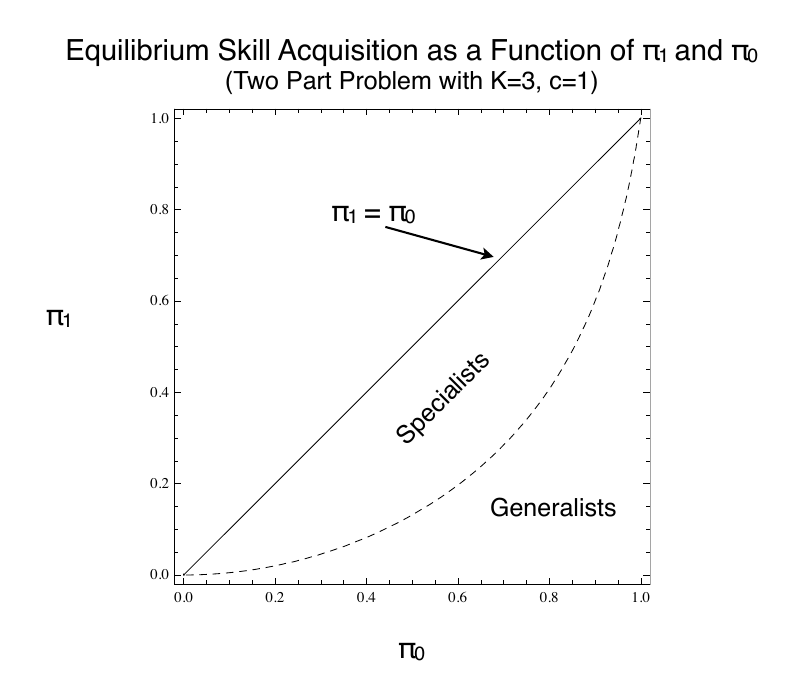}

\caption{\label{fig:Regions of Specialization in Two Part Problem}Equilibrium
skill acquisition decisions when problems are multidisciplinary, and
$k=3$ and $c=1$. }
\end{figure}

\begin{thm}
\label{Thm: Multidisciplinary Problems}If skills are independent
and symmetric within discipline, and problems multidisciplinary (eg:
$\delta=\left[\begin{array}{cc}
\delta_{1} & \delta_{0}\\
\delta_{0} & \delta_{1}
\end{array}\right]$ with $\delta_{1}>\delta_{0}$) then there is a set of values of $\pi_{1}=\delta_{1}h+\left(1-\delta_{1}\right)l$
and $\pi_{0}=\delta_{0}h+\left(1-\delta_{0}\right)l$ such that it
is individually optimal for workers to be generalists, even when there
are no barriers to solving problems in other fields. In particular,
workers will become generalists when $\left(1-\pi_{1}^{\frac{K-c}{2}}\pi_{0}^{\frac{K-c}{2}}\right)^{2}>\left(1-\pi_{1}^{K}\right)\left(1-\pi_{0}^{K}\right)$\end{thm}
\begin{proof}
WLOG, consider the case where $\delta=\left[\begin{array}{cc}
\delta_{1} & \delta_{0}\\
\delta_{0} & \delta_{1}
\end{array}\right]$ with $\delta_{1}>\delta_{0}$. A specialist in discipline $i$ will
have $K$ skills in discipline $i$. The expected probability that
her skills will solve \emph{both }parts of the problem is $E\left[P\left(\mbox{success on part 1}\right)\right]*E\left[P\left(\mbox{success on part 2}\right)\right]$,
which is $\left(1-\pi_{1}^{K}\right)\left(1-\pi_{0}^{K}\right)$ where
$\pi_{1}=\delta_{1}h+\left(1-\delta_{1}\right)l$ and $\pi_{0}=\delta_{0}h+\left(1-\delta_{0}\right)l$. 

A generalist will have skills in both disciplines. In this case, it
will be optimal for a generalist to split her skills evenly between
the two disciplines, and she will obtain $\frac{K-c}{2}$ skills in
each. The expected probability that she will solve both parts of the
problem is $\left(1-\pi_{1}^{\frac{K-c}{2}}\pi_{0}^{\frac{K-c}{2}}\right)^{2}$. 

Thus, individuals choose to generalize, when $\left(1-\pi_{1}^{\frac{K-c}{2}}\pi_{0}^{\frac{K-c}{2}}\right)^{2}>\left(1-\pi_{1}^{K}\right)\left(1-\pi_{0}^{K}\right)$
.%
\footnote{Note that when $\delta_{1}=\delta_{0}$, we have a case that fits
into the first category in the taxonomy of problem distributions in
Figure \ref{fig:Taxonomy of Distributions}. Since $\delta_{1}=\delta_{0}\implies\pi_{1}=\pi_{0}$,
we can use this calculation to verify the claim made above that in
the case where skills are symmetric, the results are the same as in
Section \ref{sec: Specialization and Barriers Between Disciplines}.%
} This region, as a function of $\pi_{1}$ and $\pi_{0}$, is illustrated
in Figure \ref{fig:Regions of Specialization in Two Part Problem}
for $K=3$ and $c=1$. The boundary of this region is defined by the
equation $\pi_{1}^{K}+\pi_{0}^{K}=\left(\pi_{1}\pi_{0}\right)^{\frac{K-c}{2}}-2\left(\pi_{1}\pi_{0}\right)^{K-c}+\left(\pi_{1}\pi_{0}\right)^{K}$
. 
\end{proof}
As would be expected, the region where individuals specialize shrinks
as the costs to generalizing ($c$) become smaller, relative to the
individual's total capacity for learning new skills $\left(M=K+c\right)$.

\section{Conclusion}

Being a generalist is costly. Every new area of expertise comes at
considerable fixed cost, in the form of a new literature, new jargon,
and new basic ideas. However, there are clearly a large (and growing)
number of individuals in research communities who choose to do so.
This raises the question of whether that decision is ever individually
rational? And is there a reason to believe that fewer people choose
to be generalists than is socially optimal?

This paper suggests that being a generalist can be a rational decision
under particular conditions. In particular, obtaining a broad range
of skills is rational if there are significant barriers to working
on questions in fields with which one is unfamiliar. Those who pay
the initial price of learning the jargon and literature of a new field
reap the benefits in the form of a larger pool of problems to solve.
It can also be rational to be a generalist if problems are multidisciplinary--that
is, if different parts of a problem are best addressed using skills
in different disciplines. Moreover, because individuals bear the costs
of becoming generalists, we will tend to have fewer of them than is
optimal from a societal standpoint. This potential market failure
means that in some cases, it is optimal for funding agencies and private
organizations to subsidize individuals in their efforts to diversify
their skills and promote interdisciplinary researchers. However, it
is unclear whether our current situation is one in which such funding
is required. More careful consideration of this question is a good
candidate for further work.

There are several elements of this model that suggest directions for
future research. It would be interesting to consider a case where
individuals differ in their innate capacity for learning skills. This
might provide some insight into what types of individuals choose to
become generalists. Incorporating collaboration would be another particularly
interesting extension. Collaboration has always been an important
part of problem solving and innovation, and it has only become more
important over time (see, among others, Laband and Tollison (2000),
Acedo et al (2006), and Goyal et al (2006)). There is reason to believe
that in a collaborative context, the advantage to generalists would
be enhanced, because generalists could connect specialists in different
fields. 

On a more general level, there is much to be gained from a better
understanding of specialization decisions. Research universities,
government organizations such as NASA, and private enterprises ranging
from Genentec to Google are reliant on the skills of individual problem
solvers. The decisions these individuals make about the breadth skills
they obtain have an undeniable effect on the rate of innovation. However,
we still have only a limited understanding of the what drives those
skill acquisition decisions, and what distinguishes the role of specialists
and generalists in problem solving. Better theoretical models of these
decisions have the potential to greatly enhance our understanding
of this important aspect of such organizations.

\section*{Appendix}

Theorem \ref{thm:If Open No Generalists-1} is the equivalent of Theorem
\ref{thm:If Open No Generalists}, and states that if individuals
can work on any available problem, then there is no advantage to being
a generalist.
\begin{thm}
\label{thm:If Open No Generalists-1}If skills are independent and
symmetric within discipline, and workers can work on any available
problem, then no worker will ever want to be a generalist and the
equilibrium population will contain only specialists.\end{thm}
\begin{proof}
As above, the \emph{ex ante }probability that a specialist in discipline
$i$ will be able to solve a problem from a given distribution, $\Delta$,
is
\begin{eqnarray*}
E\left[P\left(S_{i}\right)\right] & = & 1-\left(\delta_{i}h+\left(1-\delta_{i}\right)l\right)^{K}\\
 & = & 1-\pi_{i}^{K}
\end{eqnarray*}
 where $\pi_{i}=\left(\delta_{i}h+\left(1-\delta_{i}\right)l\right)$.

WLOG, suppose $\delta_{1}>\delta_{2}$. Since $h<l$, this means that
$\pi_{1}<\pi_{2}$ and $E\left[P\left(S_{1}\right)\right]>E\left[P\left(S_{2}\right)\right]$.
Thus, to determine whether any individual will generalize, I need
to compare $E\left[P\left(S_{1}\right)\right]$ to $E\left[P\left(G\right)\right]$.

The \emph{ex ante }probability that a generalist with $x$ skills
in discipline 1, and $K-c-x$ skills in discipline 2 solves a problem
from a given distribution, $\Delta$, is 
\begin{eqnarray*}
E\left[P\left(G\right)\right] & = & 1-\left(\delta_{1}h+\left(1-\delta_{1}\right)l\right)^{x}\left(\delta_{2}h+\left(1-\delta_{2}\right)l\right)^{K-c-x}\\
 & = & 1-\pi_{1}^{x}\pi_{2}^{K-c-x}
\end{eqnarray*}

Since $\pi_{1}<\pi_{2}$, $E\left[P\left(G\right)\right]$ is strictly
increasing in $x$. This means that a generalist will set $x=K-c-1$
and $E\left[P\left(G\right)\right]=1-\pi_{1}^{K-c-1}\pi_{2}$, which
is clearly less than $E\left[P\left(S_{1}\right)\right]=1-\pi_{1}^{K}$. 
\end{proof}
Theorem \ref{thm:Communication Barriers and Generalists-1} is a generalized
version of Theorem \ref{thm:Communication Barriers and Generalists},
and states the parameter range in which individuals will choose to
diversify their skills when there are barriers to working interdisciplinarily. 
\begin{thm}
\label{thm:Communication Barriers and Generalists-1}If skills are
independent and symmetric within discipline, and there are barriers
to working on problems in other disciplines, then there is a range
of values for $\phi$ (the fraction of problems assigned to discipline
1) for which individuals will generalize.

In particular, the ranges are as follows: 

If $\delta_{1}=\delta_{2}=\delta$, workers will obtain $K-c$ skills
spread across the two disciplines when $1-\frac{1-\pi^{K-c}}{1-\pi^{K}}\le\phi\le\frac{1-\pi^{K-c}}{1-\pi^{K}}$
, $K$ skills in discipline 1 when $\phi>\frac{1-\pi^{K-c}}{1-\pi^{K}}$,
and $K$ skills in discipline 2 when $\phi<1-\frac{1-\pi^{K-c}}{1-\pi^{K}}$. 

If $\delta_{1}>\delta_{2}$, then workers will obtain $K-c-1$ skills
in discipline 1 and one skill in discipline 2 when $1-\left(\frac{1-\pi_{1}^{K-c-1}\pi_{2}}{1-\pi_{2}^{K}}\right)\le\phi\le\frac{1-\pi_{1}^{K-c-1}\pi_{2}}{1-\pi_{1}^{K}}$
, $K$ skills in discipline 1 when $\phi>\frac{1-\pi_{1}^{K-c-1}\pi_{2}}{1-\pi_{1}^{K}}$,
and $K$ skills in discipline 2 when $\phi<1-\left(\frac{1-\pi_{1}^{K-c-1}\pi_{2}}{1-\pi_{2}^{K}}\right)$. 

If $\delta_{2}>\delta_{1}$, then workers will obtain $K-c-1$ skills
in discipline 2 and one skill in discipline 1 when $1-\left(\frac{1-\pi_{2}^{K-c-1}\pi_{1}}{1-\pi_{1}^{K}}\right)\le\phi\le\frac{1-\pi_{2}^{K-c-1}\pi_{1}}{1-\pi_{1}^{K}}$
, $K$ skills in discipline 1 when $\phi>\frac{1-\pi_{2}^{K-c-1}\pi_{1}}{1-\pi_{1}^{K}}$,
and $K$ skills in discipline 2 when $\phi<1-\left(\frac{1-\pi_{2}^{K-c-1}\pi_{1}}{1-\pi_{2}^{K}}\right)$. \end{thm}
\begin{proof}
In this case, the \emph{ex ante }probability that a problem is solved
by a specialist is $\phi\left(1-\pi_{1}^{K}\right)$ for a specialist
in discipline 1 and $\left(1-\phi\right)\left(1-\pi_{2}^{K}\right)$
for a specialist in discipline 2. Since generalists can work on problems
in both disciplines, their expected probability of solving the problem
is $1-\pi_{1}^{x}\pi_{2}^{K-c-x}$ where $x$ is the number of skills
the generalist chooses to acquire in discipline 1. First, suppose
$\delta_{1}>\delta_{2}$. Since $h<l$, this means that $\pi_{1}<\pi_{2}$
and $E\left[P\left(G\right)\right]$ is strictly increasing in $x$.
Thus, a generalist will choose a minimal number of skills in the less
useful discipline, and $E\left[P\left(G\right)\right]=1-\pi_{1}^{K-c-1}\pi_{2}$.

An individual will generalize if $E\left[P\left(S_{1}\right)\right]<E\left[P\left(G\right)\right]$
and $E\left[P\left(S_{2}\right)\right]<E\left[P\left(G\right)\right]$.
Setting $\phi\left(1-\pi_{1}^{K}\right)<1-\pi_{1}^{K-c-1}\pi_{2}$
implies that $\phi\le\frac{1-\pi_{1}^{K-c-1}\pi_{2}}{1-\pi_{1}^{K}}$.
Setting $\left(1-\phi\right)\left(1-\pi_{2}^{K}\right)<1-\pi_{1}^{K-c-1}\pi_{2}$
implies that $1-\frac{1-\pi^{K-c}}{1-\pi^{K}}\le\phi$. We can verify
that in the appropriate ranges, individuals choose to specialize.
The result follows immediately. The proof for $\delta_{2}>\delta_{1}$
is similar. For the proof when $\delta_{1}=\delta_{2}$, see Theorem
\ref{thm:Communication Barriers and Generalists}. 
\end{proof}
Finally, Theorem \ref{thm: Optimality of the Equilibrium-1} is the
generalization of Theorem \ref{thm: Optimality of the Equilibrium}.
It states that there is a parameter region in which individuals choose
to specialize, but society would prefer to have at least a few generalists. 
\begin{thm}
\label{thm: Optimality of the Equilibrium-1}If skills are independent
and symmetric within discipline, and there are barriers to working
on problems in other disciplines, then there is a range of values
for $\phi$ (the fraction of problems assigned to discipline 1) such
that generalists are underprovided in the equilibrium population of
problem solvers.

In particular, generalists are underprovided in the following ranges: 

If $\delta_{1}=\delta_{2}$, then generalists are underprovided when
$\frac{1-\pi^{K-c}}{1-\pi^{K}}<\phi<\frac{1-\pi^{N\left(K-c\right)}}{1-\pi^{NK}}$
or $1-\frac{1-\pi^{N\left(K-c\right)}}{1-\pi^{NK}}<\phi<1-\frac{1-\pi^{K-c}}{1-\pi^{K}}$ 

If $\delta_{1}>\delta_{2}$ , then generalists are underprovided when
$\frac{1-\pi_{1}^{K-c-1}\pi_{2}}{1-\pi_{1}^{K}}<\phi<\frac{1-\pi_{1}^{N\left(K-c-1\right)}\pi_{2}}{1-\pi_{1}^{NK}}$
or $1-\frac{1-\pi_{1}^{N\left(K-c-1\right)}\pi_{2}}{1-\pi_{2}^{NK}}<\phi<1-\frac{1-\pi_{1}^{K-c-1}\pi_{2}}{1-\pi_{2}^{K}}$ 

If $\delta_{2}>\delta_{1}$ , then generalists are underprovided when
$\frac{1-\pi_{1}\pi_{2}^{K-c-1}}{1-\pi_{1}^{K}}<\phi<\frac{1-\pi_{1}\pi_{2}^{N\left(K-c-1\right)}}{1-\pi_{1}^{NK}}$
or $1-\frac{1-\pi_{1}\pi_{2}^{N\left(K-c-1\right)}}{1-\pi_{2}^{NK}}<\phi<1-\frac{1-\pi_{1}\pi_{2}^{K-c-1}}{1-\pi_{2}^{K}}$ \end{thm}
\begin{proof}
First, suppose that $\delta_{1}>\delta_{2}$. The probability that
at least one of the $N$ problem-solvers in the population solves
the problem is $1-Prob\left(\mbox{none of them do}\right)$. If all
of the individuals in the population are specialists in discipline
1, then every individual has probability $\phi$ of a problem occurring
in her discipline. In that case, each specialist in discipline has
a probability $1-\pi_{1}^{K}$ of solving the problem and $\pi_{1}^{K}$
of not solving it. With probability $1-\phi$, the problem is assigned
to the other discipline, and no specialist solves it. Thus, the probability
of someone in a population of specialists solving the problem is 
\begin{eqnarray*}
Prob\left(\mbox{one of N solve it}\right) & = & 1-Prob\left(\mbox{none of N solve it}\right)\\
 & = & 1-\left[\phi Prob\left(\mbox{none solve problem in }d_{1}\right)+\left(1-\phi\right)Prob\left(\mbox{none solve problem in }d_{2}\right)\right]\\
 & = & 1-\left[\phi Prob\left(\mbox{one fails}\right)^{N}+\left(1-\phi\right)*1\right]\\
 & = & 1-\left[\phi\left(\pi_{1}^{K}\right)^{N}+\left(1-\phi\right)*1\right]\\
 & = & \phi\left(1-\pi_{1}^{KN}\right)
\end{eqnarray*}

Through a similar argument, if everyone in the population is a specialist
in discipline 2, then the probability that someone in the population
solves the problem is $\left(1-\phi\right)\left(1-\pi_{2}^{KN}\right)$.

If everyone in the population is a generalists, then the probability
of at least one person in solving the problem is 
\begin{eqnarray*}
Prob\left(\mbox{one of N solve it}\right) & = & 1-Prob\left(\mbox{none of N solve it}\right)\\
 & = & 1-\left(\pi_{1}^{K-c-1}\pi_{2}\right)^{N}\\
 & = & 1-\pi_{1}^{N\left(K-c-1\right)}\pi_{2}^{N}
\end{eqnarray*}

Society is better off with a population of generalists than a population
of discipline 1 specialists when $1-\pi_{1}^{N\left(K-c-1\right)}\pi_{2}^{N}>\phi\left(1-\pi_{1}^{KN}\right)$,
which is true when $\phi<\frac{1-\pi_{1}^{N\left(K-c-1\right)}\pi_{2}^{N}}{1-\pi_{1}^{NK}}$.
However, there is a population of generalists when $\phi\le\frac{1-\pi_{1}^{K-c-1}\pi_{2}}{1-\pi_{1}^{K}}$.
It is always the case that $\frac{1-\pi_{1}^{K-c-1}\pi_{2}}{1-\pi_{1}^{K}}\le\frac{1-\pi_{1}^{N\left(K-c-1\right)}\pi_{2}}{1-\pi_{1}^{NK}}$.
Thus, if $\frac{1-\pi_{1}^{K-c-1}\pi_{2}}{1-\pi_{1}^{K}}<\phi<\frac{1-\pi_{1}^{N\left(K-c-1\right)}\pi_{2}}{1-\pi_{1}^{NK}}$,
then society is better off with a population of generalists, but has
a population of specialists.

Through a similar argument, society is better off with a population
of generalists than a population of discipline 2 specialists when
$1-\pi_{1}^{N\left(K-c-1\right)}\pi_{2}^{N}>\left(1-\phi\right)\left(1-\pi_{2}^{KN}\right)$,
which is true when $\phi>1-\frac{1-\pi_{1}^{N\left(K-c-1\right)}\pi_{2}^{N}}{1-\pi_{2}^{NK}}$.
However, there is a population of generalists when $\phi\le1-\frac{1-\pi_{1}^{K-c-1}\pi_{2}}{1-\pi_{2}^{K}}$
. It is always the case that $1-\frac{1-\pi_{1}^{N\left(K-c-1\right)}\pi_{2}}{1-\pi_{2}^{NK}}\le1-\frac{1-\pi_{1}^{K-c-1}\pi_{2}}{1-\pi_{2}^{K}}$.
Thus, if $1-\frac{1-\pi_{1}^{N\left(K-c-1\right)}\pi_{2}}{1-\pi_{2}^{NK}}<\phi<1-\frac{1-\pi_{1}^{K-c-1}\pi_{2}}{1-\pi_{2}^{K}}$,
then society is better off with a population of generalists, but has
a population of specialists.

The proof for $\delta_{2}>\delta_{2}$ is similar. See the proof of
Theorem \ref{thm: Optimality of the Equilibrium} for the case where
$\delta_{1}=\delta_{2}$. \end{proof}

\end{document}